\documentclass[runningheads,11pt]{llncs}
\usepackage{makeidx,epsf,psfrag,amsmath,epsfig,graphics}%,amsmath}
\usepackage{algorithm}
\usepackage{algorithmic}

\usepackage[top=1in, bottom=1in, left=1in, right=1in]{geometry}

\newenvironment{keywords}{
       \list{}{\advance\topsep by0.35cm\relax\small
       \leftmargin=1cm
       \labelwidth=0.35cm
       \listparindent=0.35cm
       \itemindent\listparindent
       \rightmargin\leftmargin}\item[\hskip\labelsep
                                     \bfseries Key words:]}
     {\endlist}

\begin{document}

\setcounter{page}{1}
\pagestyle{empty}
\pagestyle{headings}

\title{A $\frac{4}{3}$-approximation algorithm for finding a spanning tree to maximize its internal vertices}

\titlerunning{A $\frac{4}{3}$-approximation Algorithm for MIST}

\author{Xingfu Li  \and Daming Zhu }

\authorrunning{Li and Zhu}

\institute{
{School of Computer Science and Technology, Shandong University,
Jinan, Shandong 250100, China. \\ Email: {\tt lxf@mail.sdu.edu.cn, dmzhu@sdu.edu.cn}}
}

\date{}
\maketitle

\begin{abstract}
    This paper focuses on  finding a spanning tree of a graph to maximize the number of its internal vertices. We present an approximation algorithm for this problem which can achieve  a performance ratio $\frac{4}{3}$ on undirected simple graphs. This improves upon the best known  approximation algorithm with  performance ratio  $\frac{5}{3}$ before. Our algorithm benefits from a new observation  for bounding the number of internal vertices of a spanning tree, which reveals that a spanning tree of an undirected simple graph has less internal vertices than the edges a maximum path-cycle cover of that graph has. We can also give an example to show that the performance ratio $\frac{4}{3}$ is actually tight for this algorithm. To decide how difficult it is for this problem to be  approximated, we show that finding a spanning tree of an undirected simple graph to maximize its internal vertices  is Max-SNP-Hard.
\end{abstract}


\section{Introduction}

The \emph{Maximum Internal Spanning Tree} problem, MIST for short, is  motivated by the designment of cost-efficient communication networks \cite{Salamon2008}.  It asks  to find a spanning tree of a graph such that the number of its internal vertices is maximized. MIST  is NP-hard,  because a Hamilton path (if present) of a graph is just a maximum internal spanning tree of that graph, and finding a Hamilton path in a graph is   NP-Hard classically \cite{Garey1979}.

%A complement problem of MIST is the so called  \emph{Minimum Leaves Spanning Tree} problem, MLST briefly. MLST asks to find a spanning tree of a graph such that the %number of its leaves is minimized. One application of MLST  appears in water resources engineering \cite{Fernau2013}.

%MLST is NP-hard and cannot be approximated to within any constant performance ratio  \cite{Lu1996}, if $P$ $\neq$ $NP$. As algorithmic approaches,  Flandrin et al. in %\cite{Flandrin2008} and  A. Kyaw in  \cite{Kyaw2009} have respectively studied the conditions of whether a graph has a spanning tree with a bounded number of leaves.

MIST is known to admit approximation algorithms with constant performance ratio.  Prieto et al. \cite{Prieto2003} first presented a 2-approximation algorithm for MIST on undirected simple graphs  by a local search technique in 2003. Later, by a slight modification, Salamon et al. \cite{Salamon2008} improved Prieto's  algorithm to running in linear-time.  Moreover, they  developed a $\frac{3}{2}$-approximation algorithm for MIST on claw-free graphs and a $\frac{6}{5}$-approximation algorithm on cubic graphs  \cite{Salamon2008}. Salamon even showed that his  algorithm in \cite{Salamon2008} can achieve  a performance ratio $\frac{r+1}{3}$  on  $r$-regular graphs  \cite{Aalamon2009PhD}. Later, Salamon \cite{Salamon2009} devised a local optimization algorithm which can approximate MIST on graphs without leaves to $\frac{7}{4}$ in $O(n^4)$ time. Through a different analysis, Knauer et al. \cite{Knauer2009} showed that  Salamon's algorithm in \cite{Salamon2009} can actually achieve  a performance ratio $\frac{5}{3}$ on undirected simple graphs in $O(n^3)$ time.  Knauer's algorithm is a simplification of the Salamon's, because  a substantially smaller neighborhood structure in the local optimization is sufficient to guarantee the approximation ratio.
%The best approximation ratio for the MIST problem on undirected simple graphs is $1.5$ which is proved by Li and Zhu in \cite{Li2014}.

Salamon et al. \cite{Salamon2009} also devoted attention to the so called weighted MIST, which asks to find a spanning tree of a  vertex weighted graph, to maximize the total weights of its internal vertices. They designed a $(2\Delta -3)$-approximation algorithm for weighted MIST on graphs without leaves with time complexity $O(n^4)$, where $\Delta$ is the maximum degree of the graph. They also proposed a  2-approximation algorithm for  weighted MIST on claw-free graphs without leaves with time complexity $O(n^4)$.  Later, Knauer et al. \cite{Knauer2009} proposed a $(3+\epsilon)$-approximation algorithm for  weighted MIST on generic undirected simple graphs.

The fix parameterized  algorithms of MIST have also been extensively studied  in the recent years.  Prieto and Sloper \cite{Prieto2003} designed the first FPT-algorithm with running  time  $O^{*}(2^{4klogk})$ in 2003. Coben et al. \cite{Coben2010} improved this algorithm to achieve a  time complexity  $O^{*}(49.4^{k})$. Then an FPT-algorithm for MIST with time complexity $O^{*}(8^{k})$\cite{Fomin2013}, and an FPT-algorithm  with time complexity $O^{*}(16^{k+o(k)})$\cite{Fomin2010} on directed graphs were proposed by Fomin et al. On  directed graphs, a randomized FPT-algorithm proposed by M. Zehavi is by now  the fastest one, which  runs in    $O^{*}(2^{(2-\frac{\Delta+1}{\Delta(\Delta-1)})k})$ time \cite{Zehavi2013}, where  $\Delta$ is the vertex degree bound of a graph.   On cubic graphs in which each vertex has degree three, Binkele-Raible et al. \cite{Fernau2013} proposed an FPT-algorithm which runs in  $O^{*}(2.1364^{k})$ time.

For the kernalization of MIST,  Prieto and  Sloper first presented an $O(k^{3})$-vertex kernel \cite{Prieto2003,Prieto2005}. Later, they improved it to  $O(k^{2})$ \cite{Prieto2005-2}.  Recently, Fomin et al. \cite{Fomin2013} gave a $3k$-vertex kernel for this problem, which is the best by now.

As for the exact exponential algorithms, Binkele-Raible et al. \cite{Fernau2013} proposed a dynamic programming algorithm for MIST with  time complexity $O^{*}(2^{n})$.  Their algorithm runs in $O^{*}((2-\epsilon)^{n})$ time on degree bounded graphs. Especially, they proposed a branching algorithm for MIST on graphs with vertex degree at most 3, which runs in $O(1.8612^{n})$ time and polynomial space.

The best performance ratio for approximating MIST has been $\frac{5}{3}$ by now \cite{Salamon2009,Knauer2009}.  Although MIST is NP-Hard, to what extent MIST rejects to be approximated has been keeping undetermined for many years.

In this paper, we devote to approximate  MIST on generic undirected simple graphs. We propose an algorithm which can approximate MIST to a performance ratio $\frac{4}{3}$. This improves upon the best known existing performance ratio  for approximating MIST \cite{Salamon2009,Knauer2009}. Primarily, our improvement is based on a new observation which reveals  that in number, those internal vertices of a spanning tree of a graph can be bounded by the edges of a maximum path-cycle cover of that graph. Thus a spanning tree can  be constructed from a maximum path-cycle cover. To arrive at a spanning tree with enough internal vertices, a graph has to be reduced by deleting some of its edges and vertices in favor of getting a maximum path-cycle cover with special natures as we cry for; then a maximum path-cycle cover has to be so reconstructed that each path component of length 1, 2 or 3 can have one of its endpoints adjacent to a vertex of a path component of length at least 4. This makes it possible to use a combinatorial way to construct a spanning tree which has three fourth times as many internal vertices as those a maximum internal spanning tree has.

For answering how difficult it is to approximate  MIST, we show that, if P $\neq$ NP, MIST rejects any polynomial time  algorithm to approximate MIST to $1+\epsilon$  for some  $\epsilon>0$. This proof is done by two reductions which  are  from (1,2)-TSP \cite{Papa1993} to the Maximum Path Cover problem \cite{Papa1993},  then from the Maximum Path Cover problem to MIST.

This paper is organized as follows. Section \ref{prliminaries} presents the concepts and notations  related to path cover, path-cycle cover, maximum internal spanning tree on graphs. Section \ref{internal-vertex-upper-bound}  presents how to bound the number of internal vertices of a spanning tree  by the number of edges of a mximum path-cycle cover. Section \ref{pre-process} presents how to reduce a graph into a special one conditioned by keeping the number of internal vertices of a maximum internal spanning tree unchanged. This just implies  that MIST on any undirected simple graph can be approximated to a performance ratio the same as that  MIST on a reduced one can be approximated to.  In section \ref{algorithm}, we devise a  $\frac{4}{3}$-approximation algorithm for MIST  on  reduced graphs. In section \ref{hardness}, we show that MIST is Max-SNP-Hard.
Section \ref{conclusion} is concluded by looking forward to the future work for MIST.

\section{Preliminaries}
\label{prliminaries}

In this paper, a graph is always undirected and simple. A path or cycle we mentioned is always simple.
Let $G$ =$(V,E)$ be an undirected simple graph. Moreover, $V(G)$ and $E(G)$ also stand for the vertex and the edge set of $G$ respectively, if there is no special emphasis.   A connected component of $G$ is a \emph{path $(resp.$ $cycle)$ component}, if it is also a path (resp. cycle) of $G$.
For $V_{1}\subseteq{V}$, a
subgraph of $G$ is  \emph{induced by} $V_1$ if it has all the vertices in $V_1$, and  the  edges of $G$ each of which has both its  ends   in $V_{1}$. The subgraph of $G$ induced by $V_1$ is abbreviated as $G[V_1]$.
A subgraph of $G$ is  a \emph{spanning subgraph} of
$G$ if it has the vertex set $V$ and  an edge set
$E_{1}$ $\subseteq$ ${E}$. The spanning subgraph of $G$ with the edge set $E_1$ $\subseteq$ $E$ is abbreviated as $G[E_1]$.

% The vertex set of $\chi$ is denoted by $V(\chi)$ where $\chi$ is a graph or a set of graphs. Similarly the edge set of $\chi$ is denoted
%by $E(\chi)$ where $\chi$ is a graph or a set of graphs.
%Two vertices in $G$ are adjacent, if $G$ has an edge which connects them.
%A vertex is  $k$-\emph{degree}, if its degree is $k$.
%For an arbitrary  graph, say $\chi$, we denote by $V(\chi)$ and $E(\chi)$ its vertex set and edge set in this paper.
%We say that two subgraphs of $G$ are \emph{{connected}} or \emph{{adjacent}} if each of them has a vertex such that these two vertices are adjacent in $G$.

A  spanning subgraph of $G$ is
 a \emph{{path-cycle cover}} of $G$ if every vertex
in it is incident with at most  $2$ edges. A path-cycle cover of $G$
is \emph{maximum} if
its edges are maximized in number over all  path
covers of $G$.

 A spanning subgraph of $G$ is
 a \emph{{path cover}} if  every
connected component of it is a path component.
A path cover of $G$ is \emph{maximum} if its edges
 are maximized in number over all  path covers of $G$.
The Maximum Path Cover problem, MPC for short, is given by an undirected simple graph, and asks to find a maximum path cover of that graph.

Although a path-cycle cover of a graph with $n$ vertices and $m$ edges can be found in $O(nm^{1.5}$ log$n)$ time \cite{Edmonds1970-1,Edmonds1970}, finding a maximum path cover of a graph is NP-Hard \cite{Papa1993}. Since a maximum path cover of a graph is also a path-cycle cover of that graph,  the size of a maximum path cover can be bounded by,

 \begin{lemma}
 \label{path-cover-bound}
 In number, a maximum path cover of a graph has no more edges than those a maximum path-cycle cover of that graph has.
 \end{lemma}

A vertex of a graph is  a \emph{leaf} if its degree  is 1, and  \emph{internal} if its degree is more than 1.
A \emph{maximum internal spanning tree} of $G$ is a spanning tree whose internal vertices are maximized in number over all  spanning trees of $G$.
The \emph{Maximum Internal Spanning Tree} problem, MIST namely, is given by an undirected simple graph, and asks to find a maximum internal spanning tree of that graph.

\section{A bound for the number of internal vertices of a spanning tree}
\label{internal-vertex-upper-bound}
Let $G$ =$(V,E)$ be an undirected connected simple graph.
In this section, we show that  a spanning tree of $G$ has less internal vertices than the edges a maximum path-cycle cover  of $G$ has.

%In what follows a graph has more than one vertex.
\begin{lemma}
\label{tree-path-number}
If a tree has more than one vertex,  then in number, it has a path cover  which has less path components than those  leaves it has.
\end{lemma}

\begin{proof}
   % Let $T$ be a tree with $x$ $(\geq 2)$ leaves. We  show that there must be a path cover of $T$ with at most $x-1$ paths.

Let $T$ be a   tree with $x$ $>$ $1$ leaves. The proof is an inductive method on $x$.
If $x=2$, $T$  is a path component, the lemma holds true of course. Then the inductive assumption is, if a tree has at most $x-1$ leaves, it has a path cover which has less path components than the  leaves it has. Later, we show that if $T$ has $x$ $(>2)$ leaves, it must have a path cover with at most $x-1$ path components.

Since $x>2$, a path, say $P$ $=$ $u$, ..., $v$ $\neq$ $u$, can be identified in $T$, where $u$ and $v$ are both  leaves of $T$. We then delete those edges incident to the  vertices  of $P$ except those $P$ has. This gives rise to a spanning forest of $T$. Let $T_{1}$, ..., $T_{j}$, $T_{j+1}$, ..., $T_{k}$ be all the trees in the forest except $P$, where  $T_{i}$ for $1\leq i\leq j$ has only one vertex while the others  do not. Note that the vertex in $T_{i}$ for $1\leq i\leq j$ is also a leaf of $T$. Namely, one path can cover $T_{i}$ for  $1$ $\leq$ ${i}$ $\leq$ $j$. Moreover, $T_{i}$ for  $j+1$ $\leq$ ${i}$ $\leq$ $k$ has at most $x-1$ leaves because in addition to rejecting leaves $u$ and $v$, it has at most one leaf which does not act as  a leaf in $T$.  Let  $T_{j+1}$, $T_{j+2}$, ..., $T_{k}$ have $x_{j+1}$, $x_{j+2}$, ..., $x_{k}$ leaves  respectively. By the inductive assumption, $T_{i}$ for  $j+1\leq i\leq k$  must have a path cover with at most $x_{i}$ $-$ $1$ path components. Hence $T$ has a path cover with  at most $1$ $+$ $j$ $+$ $\sum_{j+1\leq i\leq k}(x_{i}-1)$ $\leq$ $1+j$ $+$ $\sum_{j+1\leq i\leq k}(x_{i})$ $-$ $(k-j)$ $\leq$ $1$ $+$ $j$ $+$ $(x-(2+j)+(k-j))$ $-$ $(k-j)$ $=$ $x-1$ path components. \qed
\end{proof}

\begin{theorem}
\label{internal-bound}
In number, a maximum internal spanning tree of $G$ has  less internal vertices than the  edges  a maximum path-cycle cover of $G$ has.
\end{theorem}

\begin{proof}
Let $P^*$ be a maximum path cover of $G$ with $|P^*|$ path components. By the fact $|P^*|$ $+$ $|E(P^*)|$ $=$ $|V|$, those path components in $P^*$ have just $|V|$ $-$ $|P^*|$ edges. If $|P^*|$ $=$ $1$, then a maximum internal spanning tree of $G$ is a Hamilton path, the proof is trivial. Later, let $|P^*|$ $>$ $1$. If a spanning tree of $G$ has at least  $|E(P^*)|$ internal vertices, it must have at most $|P^*|$ leaves. Then by Lemma \ref{tree-path-number}, we can find a path cover of this tree with at most $|P^*|$ - 1 path components, or in other words, with at least $|E(P^*)|$ +1 edges. Thus, $G$ also has a path cover with $|E(P^*)|$ +1 edges, which means $P^*$ is not maximum, a contradiction.

By Lemma \ref{path-cover-bound}, the proof is done. \qed
\end{proof}
\iffalse
\begin{corollary}
\label{using-internal-bound}
In number, a maximum internal spanning tree of $G$ has less internal vertices  than  the edges a maximum path-cycle cover of $G$ has.
\end{corollary}
\fi

Due to Theorem \ref{internal-bound}, a simple algorithm arises to approximate MIST to a performance ratio 2: (1)find a maximum path-cycle cover of $G$, say $H$, in which  each  cycle component has at least four edges; (2)delete one edge from each cycle component in  $H$ to transform $H$ into $H'$ as a path cover of $G$; (3)link all path components in $H'$ into a spanning tree of $G$ by adding edges of $G$ to $H'$, where step (3) works because $G$ defaults to be connected. This is  a 2-approximation algorithm, because it results in a spanning tree of $G$ to which,  each cycle component of $H$ with $k$ $\geq$ $4$ edges contributes at least $k-2$ internal vertices; each path component of length one or two  contributes at least one internal vertex; each path component of length $k$ $\geq$ $3$  contributes at least $k-1$ internal vertices.

To ensure  a better performance ratio to approximate MIST, it is necessary to make those connected components in a maximum path-cycle cover contribute more internal vertices to that spanning tree to be constructed.  To ensure a performance ratio $\frac{4}{3}$, we have to reduce $G$ by deleting some of its edges and vertices , which will be stated in the next section.

\section{Edge and vertex reducing}
\label{pre-process}
Let $G$ = $(V,E)$ be  an  undirected connected simple graph. Deleting an edge of $G$ refers to removing that edge from $G$; deleting a vertex of $G$ refers to removing that vertex and the boundary edges from $G$. A deletion of an edge (or a vertex) of $G$ is \emph{safe}, if the deletion results in a subgraph
of $G$ which has a maximum internal spanning tree with no less internal vertices than those  a maximum internal spanning tree of $G$ has. Only by deleting some edges and vertices of $G$ safely, can we link those components in a maximum path-cycle cover into  a tree with as many internal vertices as we want.

 Two vertices are \emph{adjacent}, if they are both incident to one edge. Two vertices are adjacent \emph{respecting an edge}, if they are both incident to that edge.
 An edge of a graph is referred to as a \emph{cut edge} if deleting it can result in more connected components than those in that graph.

\begin{lemma}
\label{delete-edge-property}
%If in $G$, an edge other than a cut edge is incident with two vertices each of which  is adjacent to a leaf,  then the deletion of the edge  is safe.
If in $G$, an edge has both its ends  adjacent to  leaves respectively,  and is not a cut edge,  then the deletion of the edge is safe.
\end{lemma}

\begin{proof}
Let $(u,v)$ be an edge other than a cut one of $G$, where $u$ and $v$ are adjacent to leaves respectively.
Let $T^{*}$ be a maximum internal spanning tree of $G$.  If  $(u,v)$ $\notin$ $E(T^{*})$, then the proof is done. Later, let $(u,v)$ $\in$ $E(T^{*})$.  Let $T_{u}$, $T_{v}$ be those two sub trees of $T^*$ resulted by removing  $(u,v)$ from $T^{*}$, where  $u$ $\in$ $V(T_{u})$, $v$ $\in$ $V(T_{v})$. Note that a leaf of $G$ must be a leaf of $T^*$. Thus $u$ and $v$ must be internal in $T^{*}$,  otherwise, $T^*$ is not connected.   Since $(u,v)$ is not a cut edge of $G$, there must be  an edge $(u^{'},v^{'})$ $\in$ $E(G)$ $\setminus$ $\{(u,v)\}$ with $u^{'}$ $\in$ $V(T_{u})$ and $v^{'}$ $\in$ $V(T_{v})$. Then removing  $(u,v)$ from $T^{*}$, and adding  $(u^{'},v^{'})$ to it must result in a spanning tree of $G$ as well as
$G[E(G)$ $\setminus$ $\{(u,v)\}]$, which can be denoted as $G[(E(T^*)$ $\setminus$ $\{(u,v)\})$ $\cup$ $\{(u',v')\}]$.
The vertices $u$ and $v$ are both internal in this spanning tree because,

$(1)$If $u$ = $u'$, then $v$ $\neq$ $v'$. Thus, $v$ is internal in $T_v$, and moreover  in $G[(E(T^*)$ $\setminus$ $\{(u,v)\})$ $\cup$ $\{(u',v')\}]$. Adding $(u',v')$ to $T^*[E(T^*)$ $\setminus$ $\{(u,v)\}]$ must make $u$ to be internal in  $G[(E(T^*)$ $\setminus$ $\{(u,v)\})$ $\cup$ $\{(u',v')\}]$.

$(2)$If $u$ $\neq$ $u'$, then $u$ is internal in $T_u$, and moreover in $G[(E(T^*)$ $\setminus$ $\{(u,v)\})$ $\cup$ $\{(u',v')\}]$. No matter whether $v$ = $v'$ or not,  adding $(u',v')$ to $T^*[E(T^*)$ $\setminus$ $\{(u,v)\}]$ must make $v$ to be internal in  $G[(E(T^*)$ $\setminus$ $\{(u,v)\})$ $\cup$ $\{(u',v')\}]$.  \qed
\end{proof}

Repeating the safe edge deletion stated by Lemma \ref{delete-edge-property} until no edge exists to subject to Lemma \ref{delete-edge-property}, must transform $G$ into a subgraph of $G$ which  subjects to the following lemma.

\begin{corollary}
\label{delete-edges}
There is a subgraph of $G$ in which, $(1)$an edge must be a cut edge, if its two ends both are adjacent to leaves respectively; $(2)$a maximum internal spanning tree of it has no less internal vertices than those  a maximum internal spanning tree of $G$ has.
\end{corollary}

\iffalse
\begin{proof}
Repeating the safe edge deletion stated by Lemma \ref{delete-edge-property} until no edge exists to subject to Lemma \ref{delete-edge-property}, will transform $G$ into a subgraph of $G$ which  subjects to $(1)$ and  $(2)$.  \qed
\end{proof}
\fi

A vertex is referred to as a \emph{cut vertex} of a graph, if deleting it  results in more connected components than those in that graph.    A cut vertex of a graph is \emph{super}, if deleting it results in at least 2 more connected components than those in that graph. Since $G$ defaults to be connected, deleting a super cut vertex of $G$ must result in at least 3 connected components.  For identifying the safe deletions of vertices of $G$, we concentrate on those leaves which are adjacent to  super cut vertices.

\begin{lemma}
\label{leaf-delete-property}
If in $G$, a leaf is adjacent to a super cut vertex, then the deletion of it is safe.
\end{lemma}

\begin{proof}
Let $T^{*}$ be a maximum internal spanning tree of $G$. Let $u$ be a super cut vertex of $G$, while $v$ be a leaf adjacent to $u$ respecting an edge of $G$. Then the degree of $u$ in $T^{*}$ is  at least three. Namely, deleting $u$ from $T^{*}$ must yield at least three connected components.  Since $v$ is a leaf of $T^{*}$,  $v$ can be deleted from $T^*$ with $u$ as an internal vertex of $T^*[V\setminus\{v\}]$, where $T^*[V\setminus\{v\}]$ is  a spanning tree of $G[V\setminus\{v\}]$.   \qed
\end{proof}

Lemma \ref{leaf-delete-property} indicates that  those leaves of a graph which are adjacent to super cut vertices have no contribution for finding a maximum internal spanning tree. Thus, by the safe deletions of edges and leaves, we can get a graph as stated in,

\begin{corollary}
\label{delete-leaves}
There is a subgraph of $G$ which subjects to,
\begin{itemize}
    \item[$(1)$] an edge must be a cut edge, if its two ends each is adjacent to a leaf;
    \item[$(2)$] a cut vertex is not super, if it is adjacent to a leaf;
    \item[$(3)$] A maximum internal spanning tree of the subgraph has no less internal vertices than those  a maximum internal spanning tree of  $G$ has.
\end{itemize}
\end{corollary}

\begin{proof}
By Corollary \ref{delete-edges}, let $G_1$ be a subgraph of $G$ which subjects to Item (1) and (3) of the corollary.
If a leaf is adjacent to a super cut vertex of $G_1$, then by Lemma \ref{leaf-delete-property}, it can be deleted from $G_1$. This leaf deletion for $G_1$ must result in a subgraph of $G_1$, which also subjects to Item (1) and (3), because the deletion is safe, and moreover, does not bring in any new cut edge to $G_1$, and take away any  existing cut edge with two ends adjacent to leaves from $G_1$. Let $G_2$ be a subgraph of $G_1$ resulted by repeating such operation, until no  leaf can be found for deletion. Then, $G_2$ must subject to  Item $(1)$, $(2)$ and $(3)$. \qed
\end{proof}

%Let $G_1$ be a subgraph of $G$ for which Corollary \ref{delete-leaves} holds true. Since $G$ is not a tree, every internal vertex of $G_1$ is adjacent to at most one %leaf. %Otherwise, let $u$ be an internal vertex of $G_1$ which is adjacent to at leat two leaves. Then $u$ is a super cut vertex of $G_1$ and is adjacent to a leaf. This %contradicts to the statement (1) of Corollary \ref{delete-leaves}.

%
%By Corollary \ref{delete-edges} and  Corollary \ref{delete-leaves}, $G$ must be reduced into a graph which subjects to Corollary \ref{delete-leaves}.
We directly name by Reduce$(G)$ the algorithm for $G$ to delete its edges and vertices safely by the methods in Corollary \ref{delete-edges} and  \ref{delete-leaves}, without formalizing its details.

Recall that $G$ =$(V,E)$. It takes $O(|V|$ $+$ $|E|)$ time to decide whether an edge of $G$ has two ends adjacent to respective leaves. Moreover, it takes   $O(|V|+|E|)$ time to decide whether an edge is a cut one. So completing the deletions of edges for $G$ takes  $O(|E|(|V|+|E|))$ time. It takes  $O(|V|+|E|)$ time to decide whether a vertex of $G$ is a super cut vertex, or whether it is adjacent to a leaf. There are at most $|V|$ leaves to be deleted. So completing the deletions of vertices  takes $O(|V|(|V|+|E|))$ time. To sum up, the time complexity of  Reduce$(G)$ is $O((|V|+|E|)^{2})$.

A subgraph of  $G$ is \emph{reduced} if it subjects to Corollary \ref{delete-leaves}. A reduced subgraph of $G$ must have the same set of internal vertices as $G$ has. Moreover, every internal vertex of a reduced graph is adjacent to at most one leaf. By Corollary \ref{delete-edges} and \ref{delete-leaves}, Reduce$(G)$  must return a reduced subgraph of $G$. In the following, we show that it suffices to approximate MIST on a reduced subgraph of $G$ for approximating MIST on $G$. Actually, each   spanning tree of a reduced subgraph of $G$ can turn into a spanning tree of $G$ with those internal vertices unchanged. That is,
\begin{lemma}
\label{alg-internal-number}
For each spanning tree of a reduced subgraph of $G$, $G$ has a spanning tree  which has the same set of internal vertices as  the spanning tree of that reduced subgraph of $G$ has.
\end{lemma}

\begin{proof}
Let $G_1$ be a reduced subgraph of $G$, while $T_1$ be a spanning tree of $G_1$. Since $G_1$ has the same set of internal vertices as $G$, a spanning tree of $G$ can be obtained by adding to $T_1$ the leaves in $V(G)$ $\setminus$ $V(G_1)$. Such a spanning tree of $G$ must have the same set of internal vertices as that  $T_1$ has. \qed
\end{proof}

Let $G_{1}$ be a reduced subgraph of $G$. Let $I(G)$, $I(G_{1})$ be the sets of internal vertices of the maximum internal spanning trees of $G$ and $G_{1}$ respectively. Since $G_1$ is a subgraph of $G$, $|I(G)|$ $\geq$ $|I(G_1)|$. Then $|I(G)|$ $=$ $|I(G_1)|$ follows from Corollary \ref{delete-leaves}.  If $T_{1}$ is a  spanning tree of $G_{1}$ with $I(T_{1})$ as its set of internal vertices, then a spanning tree of $G$, say $T$ can be made from $T_1$ by adding the leaves in $V(G)$ $\setminus$ $V(G_1)$ to $T_1$. By Lemma \ref{alg-internal-number}, $|I(T)|$ = $|I(T_1)|$. It follows that $\frac{|I(G)|}{|I(T)|}$ $=$ $\frac{|I(G_{1})|}{|I(T_{1})|}$. In other words, if MIST can be approximated to a substantial performance ratio on reduced graphs,  so can MIST be done  on undirected simple graphs. In the next section, we focus on reduced graphs to ask for their spanning trees.

\section{How to find a spanning tree in a reduced graph}
\label{algorithm}
In this section, $G_{1}$ always stands for  a connected reduced graph instead of a tree.
Ordinarily, a maximum path-cycle cover can be found in $O(n^{2}m)$ time in an undirected simple graph, even if each cycle component is restricted  to have at least 4 edges \cite{Hartvigsen1984}, where $n$ = $|V(G_1)|$ and $m$ = $|E(G_1)|$.
 We focus on finding a spanning tree of $G_1$ with at least $\frac{3}{4}$ times as   many internal vertices as those a maximum internal spanning tree of $G_1$ has. By Theorem \ref{internal-bound},  it suffices to construct a spanning tree of $G_1$ which has at least $\frac{3}{4}$ times as many internal vertices as the edges a maximum path-cycle cover of $G_1$ has. To hit this point, we try to reconstruct the maximum path-cycle cover of $G_1$ at first.

\subsection{Reconstruction of a maximum path-cycle cover}

We also treat a maximum path-cycle cover as a set of cycle components and path components. The reconstruction aims to transform a maximum path-cycle cover of $G_1$ into such one  that each path component can contribute as many internal vertices as its edges to that spanning  tree to be constructed, if its length is no larger than three.
A path component is  a \emph{singleton} if its length is zero.
A vertex of a path component is  \emph{inner} if its degree in it is $2$, and an \emph{endpoint} otherwise.  The \emph{endpoint} of a singleton is
the singleton itself.

Note that although two vertices in distinct connected components in a maximum
 path-cycle cover cannot be adjacent respecting any edge of the maximum path-cycle cover,
 they can be \emph{adjacent} respecting an edge of $G_{1}$.
If a maximum path-cycle cover of $G_1$ contains only one connected component, then getting a maximum internal spanning tree of $G_1$ is trivial.
Thus in what follows, a maximum path-cycle cover of $G_1$ is  assumed to have more than one connected component.   Since $G_1$ defaults to be connected, every connected component in a maximum path-cycle cover must have at least one vertex adjacent to a vertex outside it respecting an edge of $G_1$.

A maximum path-cycle cover can be transformed into one, in which each path component has one endpoint adjacent to a vertex outside it if its length is no more than 3. Those path components of length at most 2 can be dealt with into one as the following lemma states.

\begin{lemma}
\label{path-of-length-no-larger-than-2-property}
There exists such a maximum path-cycle cover of $G_1$  that, if a path component is of length no larger than $2$, then  it has one endpoint adjacent to a vertex outside it.

\end{lemma}
\begin{proof}
    Let $H$ be a maximum path-cycle cover of $G_1$. Let $p$ be a  path component of length no larger than 2 in $H$.
    If  $p$ is a singleton or has one edge, one endpoint of $p$ must be adjacent to a vertex outside $p$ respecting an edge of $G_1$, because $G_{1}$ is connected.

    If the length of $p$ is 2, the endpoints of $p$ cannot both be leaves of $G_{1}$, because if so, $G_{1}$ is a tree, or not reduced. If $p$ has just one endpoint as a leaf of $G_1$, then the other endpoint of $p$ must be adjacent to a vertex outside $p$ respecting an edge of $G_1$, because $G_1$ is connected and simple.

    If the length of $p$ is 2, and either of the two endpoints of $p$ is not a leaf of $G_1$ and  not adjacent to any vertex outside  $p$ respecting an edge of $G_1$, then the two endpoints of  $p$ are adjacent respecting an edge of $G_1$. In this situation, $p$ can be replaced by another path component of length 2. Concretely, let $p=$ $v_1$ $v_2$ $v_3$, then $(v_1,v_3)$ must be an edge of $G_1$. Thus $q$ = $v_3$ $v_1$ $v_2$ is also a path component of  length 2. Since $G_1$ is connected, as one endpoint of $q$, $v_2$ must be  adjacent to a vertex outside $q$ respecting an edge of $G_1$. So $H\setminus\{p\}\cup\{q\}$ is also a maximum path-cycle cover of $G_1$. Such kind of replacement can be done for every  path component of length 2, if it has two endpoints adjacent to each other but adjacent to no vertex outside it respecting an edge in $G_1$. When no path component of length 2 can be replaced, $H$ must be transformed into a maximum path-cycle cover as what the lemma states.    \qed
\end{proof}

To  deal with those path components of length three, we have to exclude a situation where two endpoints of a path component are both leaves of $G_1$.

\begin{lemma}
\label{no-K1-3-path}
In a maximum path-cycle cover of $G_1$, if a path component has three edges, then its two endpoints are not both leaves of $G_1$.
\end{lemma}

\begin{proof}
 Let $p=$ $u_{1}$ $u_{2}$ $u_{3}$ $u_{4}$ be a path component of length three in a maximum path-cycle cover, and $u_{1}$, $u_{4}$ be leaves of $G_1$.   By Corollary \ref{delete-leaves}, $(u_{2},u_{3})$ is a cut edge of $G_1$. Since a maximum path-cycle cover of $G_1$ has at least two connected components, either  $u_{2}$ or $u_{3}$ must be adjacent to a vertex outside $p$ respecting an edge of $G_1$. Without loss of generality, let $u_{2}$ be  adjacent to a vertex outside $p$ respecting an edge of $G_1$. So the deletion of  $u_{2}$ from $G_1$ will yield at least three connected components because $u_{1}$ is a leaf and $(u_{2},u_{3})$ is a cut edge.  This comes to a contradiction to the assumption that $G_1$ is reduced. \qed
\end{proof}

By the following two lemmas, we  show that  a path component of length three has one endpoint adjacent to a vertex outside it respecting an edge of $G_1$, or can be transformed into one which has one endpoint adjacent to a vertex outside it respecting an edge of $G_1$, no matter whether that path component has an endpoint acting as a leaf of $G_1$ or not,

 \begin{lemma}
 \label{no-special-3-path}
In a maximum path-cycle cover of $G_1$, if a path component of length three has no endpoint as a leaf of $G_1$, then it must have an endpoint adjacent to a vertex outside it.
 \end{lemma}

  \begin{proof}
  Let  $p=$ $u_{1}$ $u_{2}$ $u_{3}$ $u_{4}$ be a path component of length three in a maximum path-cycle of  $G_1$, where $u_{1}$ and $u_{4}$ are endpoints of it rather than leaves of $G_1$. If neither $u_{1}$ nor $u_{4}$ is adjacent to any vertex outside $p$ respecting an edge of $G_1$, there must be two vertices $v$, $v^{'}$ $\in$ $\{u_{2}$, $u_{3}\}$  such that $(u_{1},v)\in E(G_1)$ and $(u_{4},v^{'})\in E(G_1)$.  This leads to a contradiction because,

  (1) $v$ $\neq$  $u_{2}$ and $v^{'}$ $\neq$  $u_{3}$,  otherwise, $G_1$ is not simple.

  (2) $v$ $=$ $u_{3}$ and $v^{'}$ $=$ $u_{2}$ can not happen simultaneously, because if so, $u_{1}$ $u_{3}$ $u_{4}$ $u_{2}$ will form a cycle component of length four, which contradicts to the assumption that  $p$  belongs to a maximum path-cycle cover.
\qed
  \end{proof}

\begin{lemma}
\label{3-path-with-one-endpoint-as-a-leaf}
In a maximum path-cycle cover of $G_1$, if a path component of length three has just one endpoint as a leaf of $G_1$, then respecting an edge of $G_1$, it has one endpoint adjacent to a vertex outside it, or can be transformed into one with one endpoint adjacent to a vertex outside it.
\end{lemma}

\begin{proof}
Let $p=u_1$ $u_2$ $u_3$ $u_4$ be a path component of length three in a maximum path-cycle cover of $G_1$. Without loss of generality, let $u_4$ be a leaf of $G_1$.  If $u_1$ is adjacent to a vertex outside $p$ respecting an edge of $G_1$,  the proof is done. Otherwise, $u_1$ must be adjacent to $u_3$ respecting an edge of $G_1$ because $u_1$ is not a leaf of $G_1$ and $G_1$ is simple. Since $p$ is not the unique component in the maximum path-cycle cover, $u_2$ or $u_3$ must be adjacent to a vertex outside $p$ respecting an edge of $G_1$.

(1) If $u_3$ is adjacent to a vertex outside $p$ while $u_2$ is not, then deleting $u_3$ from $G_1$ will yield at least three connected components. Thus, $u_3$ is a super cut vertex of $G_1$ and adjacent to a leaf, which means $G_1$ is not reduced, a contradiction. That is, $u_3$ cannot be adjacent to any vertex outside $p$.

(2) If $u_2$ is adjacent to a vertex outside $p$, then $p^{'}=u_2$ $u_1$ $u_3$ $u_4$ is a path component of length three with $V(p')$ $=$ $V(p)$. Replacing $p$ with $p^{'}$ in the maximum path-cycle cover, $p$ is transformed into a path component of length three with one end point adjacent to a vertex outside it respecting an edge of $G_1$.
\qed
\end{proof}

Summing up the reconstructions for a maximum path-cycle cover,  we have,
\begin{lemma}
\label{path-of-length-three-property}
There is such a maximum path-cycle cover of $G_1$ that every path component of length at most $3$ must have an endpoint adjacent to a vertex outside it respecting an edge of $G_1$.
\end{lemma}

\begin{proof}
Let $H$ be a maximum path-cycle cover of $G_1$. By  Lemma \ref{path-of-length-no-larger-than-2-property}, every path component of length at most 2 in $H$ has one endpoint adjacent to a vertex outside it or can be transformed into one  which  has one endpoint adjacent to a vertex outside it respecting an edge of $G_1$.

By  Lemma \ref{no-K1-3-path}, \ref{no-special-3-path} and \ref{3-path-with-one-endpoint-as-a-leaf},  every path component of length three in $H$ has one endpoint adjacent to a vertex outside it, or can be transformed into one which has one endpoint adjacent to a vertex outside it respecting an edge of $G_1$. That is all for the proof. \qed
\end{proof}

If a path component has one endpoint adjacent to  a vertex of a cycle component respecting an edge of $G_{1}$, they  can be merged  into one path component. Thus,

\begin{lemma}
\label{path-cycle-merge}
There exists a maximum path-cycle cover of $G_{1}$ in which no endpoint of a path component is adjacent to a vertex of any cycle component respecting an edge of $G_{1}$.
\end{lemma}

\begin{proof}
In a maximum path-cycle cover, let $p$ be a path component which has an endpoint, say $u$, adjacent to a vertex, say $v$, of a cycle component, say $c$, respecting an edge of $G_1$. Then adding  $(u,v)$ and removing an edge incident to $v$ of $c$ will merge $p$ and $c$ into  a path component of length $|E(p)|$ + $|E(c)|$. This can be done for every path component with one endpoint adjacent to a vertex of a cycle component respecting an edge of $G_1$, which must result in a  maximum path-cycle cover as what the lemma states. \qed
\end{proof}

Recall that we look toward arriving at a tree to which each path component can contribute as many internal vertices as the edges it has, if its length is no more than three. To meet this aim, we will  reconstruct the maximum path-cycle cover into one, in which each path component of length 1, 2 or 3 has one endpoint adjacent to an inner vertex of a path component of length at least four.  We have to  deal  with those path components of length 1 beforehand.

\begin{lemma}
\label{property-of-path-less-than-four-1}
There is such a maximum path-cycle cover of $G_1$ that, respecting an edge of $G_1$, each path component of length $1$ has an endpoint adjacent to an inner vertex of a path component of length at least $3$.
\end{lemma}
\begin{proof}
Let $H$ be a maximum path-cycle cover of $G_1$ which subjects to  Lemma \ref{path-of-length-three-property} and  \ref{path-cycle-merge}.
If respecting an edge of $G_1$, one endpoint of a path component of length 1 is adjacent to an inner vertex of another path component of length 2,  then these two path components  can be replaced by a singleton and another path component of length  three which has one endpoint adjacent to a vertex outside it. Concretely, let $p$ $=$ $u_{1}$ $u_{2}$ be a path component of length 1, $q=$ $v_{1}$ $v_{2}$ $v_{3}$ be a path component of length 2. Without loss of generality, let $u_1$ be adjacent to $v_2$ respecting an edge of $G_1$. By Lemma \ref{path-of-length-three-property}, let  $v_3$ be adjacent to a vertex other than $v_1$, $v_2$. Moreover, $v_3$ cannot be adjacent to $u_1$ or $u_2$, because if so, $H$ cannot be  maximum.   Then in $H$, $p$ and $q$ can be replaced by $p'$ = $v_{3}$ $v_{2}$ $u_{1}$ $u_{2}$ and $q^{'}$ = $v_1$ with  $|E(p)|$ $+$ $|E(q)|$ $=$ $|E(p^{'})|$ $+$ $|E(q^{'})|$. Since $v_3$ is adjacent to another vertex than $u_1$, $u_2$, $v_1$, $v_2$, $H$ $\setminus$ $\{p,q\}$ $\cup$ $\{p^{'}$, $q^{'}\}$ must be also a maximum path-cycle cover of $G_1$ which subjects to Lemma \ref{path-of-length-three-property} and  \ref{path-cycle-merge}. Repeating this operation  if respecting an edge of $G_1$, a path component of length 1 has one endpoint adjacent to an inner vertex of a path component of length 2,  will transform $H$ into a maximum path-cycle cover of $G_1$ as stated in this lemma.
\qed
\end{proof}

\begin{lemma}
\label{property-of-path-less-than-four}
There is such a maximum path-cycle cover of $G_1$ that, respecting an edge of $G_1$, $(1)$ every singleton is adjacent to an inner vertex of a path component; every path component of length $1$, $2$ or $3$ has an endpoint adjacent to an inner vertex of a path component of length at least four.
\end{lemma}

\begin{proof}
Let $H$ be a maximum path-cycle cover of $G_1$ which subjects to  Lemma \ref{path-of-length-three-property},  \ref{path-cycle-merge} and \ref{property-of-path-less-than-four-1}.

$(1)$A singleton cannot be adjacent to an endpoint of any path component in $H$ respecting an edge of $G_1$, because $H$ is  maximum. A singleton cannot be adjacent to any vertex of a cycle component in $H$ respecting an edge of $G_1$ by Lemma \ref{path-cycle-merge}. Thus, a singleton must be adjacent to an inner vertex of a path component respecting an edge of $G_1$.

%Then by Lemma \ref{path-cycle-merge}, any path component in $H$ has no endpoints adjacent to any vertex of a cycle component respecting an edge of $G_1$.
$(2)$If in $H$, every path component of length 1, 2 or 3 has an endpoint adjacent to an inner vertex of a path component of length at least four respecting an edge of $G_1$, the proof is done. Otherwise,   we can transform $H$ into such a maximum path-cycle cover as what the lemma states. Let $p$  be a path component of length 1, 2 or 3, while $q$ be a path component of length 2 or 3, such that one endpoint of $p$ is adjacent to an inner vertex of  $q$.
By Lemma \ref{property-of-path-less-than-four-1},  $p$ and $q$ have at least four edges. Thus it suffices to show that $p$ and $q$ can always be replaced by  a singleton and a path component of length $|E(p)|$ $+$ $|E(q)|$.

Let $p=$ $u_{1}$ $x_1$ $x_2$ $u_{2}$, $q$ = $v_1$ $v_2$ $y$ $v_3$,  where $x_1$, $x_2$, $y$ may be nonexistent.  Let $(u_{1}, v_{2})$ be an edge of $G_1$.  Then we can replace $p$ and $q$ by $p'$ = $v_{3}$ $y$ $v_{2}$ $u_{1}$ $x$ $u_{2}$ and $q^{'}=$ $v_{1}$, where  $|E(p)|$ $+$ $|E(q)|$ $=$ $|E(p^{'})|$ $+$ $|E(q^{'})|$. Thus, $H$ $\setminus$ $\{p,q\}$ $\cup$ $\{p^{'}$, $q^{'}\}$ is also a maximum path-cycle cover of $G_1$. Since $|E(p)|$ + $|E(q)|$ $\geq$ $4$, this replacement must eliminate two path components of length 1, 2 or 3 in $H$. We insist to denote by  $H$  the maximum path-cycle cover resulted by replacing $\{p,q\}$ with $\{p',q'\}$ in $H$. Then by Lemma \ref{path-cycle-merge}, repeating this replacement in $H$ if a path component of length 1, 2 or 3 has one endpoint adjacent to an inner vertex of length 2 or 3, will transform $H$ into a maximum path-cycle cover, in which  each path component of length 1, 2 or 3 in $H'$ has one endpoint adjacent to an inner vertex of a path component of length at least 4 resecting an edge of $G_1$.
 \qed
\end{proof}

\iffalse
Summing up the reconstructions for those path components of length 1, 2 or 3, a maximum path-cycle cover of $G_1$ can be transformed into one which subjects to,

\begin{lemma}
\label{new-H-property}
There is such a maximum path-cycle cover of $G_1$ that, respecting an edge of $G_1$, $(1)$ every singleton is adjacent to an inner vertex of a path component; $(2)$ every path component of length $1$, $2$ or $3$ has an endpoint adjacent to an inner vertex of a path component of length at least four.
\end{lemma}

\begin{proof}
Let $H$ be a maximum path-cycle cover of $G_1$ which subjects to Lemma  \ref{path-cycle-merge} and \ref{property-of-path-less-than-four}.

(1)A singleton cannot be adjacent to an endpoint of any path component in $H$ respecting an edge of $G_1$, because $H$ is  maximum. A singleton cannot be adjacent to any vertex of a cycle component in $H$ respecting an edge of $G_1$ by Lemma \ref{path-cycle-merge}. Thus, a singleton must be adjacent to an inner vertex of a path component respecting an edge of $G_1$.

(2)Since $H$ subjects to Lemma \ref{property-of-path-less-than-four},  a path component of length 1, 2 or 3 must  have  one  endpoint adjacent to an inner vertex of a path component of length at least four respecting an edge of $G_1$. \qed
\end{proof}
\fi
In the next subsection, we start with a maximum path-cycle cover which subjects to Lemma \ref{property-of-path-less-than-four}, to assemble the connected components in it into a spanning tree.

\subsection{Assemble of a spanning tree}

In this subsection,
let $H$ be a  maximum path-cycle cover of  $G_{1}$  which subjects to Lemma \ref{property-of-path-less-than-four}. Let   $T$ be a subtree of $G_1$. A path component, say  $p$ in $H$,
\emph{joins}  $T$, if   $V(p)$ $\subseteq$ $V(T)$ and $E(p)$ $\subseteq$ $E(T)$. A cycle component, say $c$ in $H$,  \emph{joins}  $T$, if  $V(c)$ $\subseteq$ $V(T)$ and $|E(c)\cap{E(T)}|$ $\geq$ $|E(c)|-1$.  We specially pay attention to
those subtrees of $G_1$ which are joined by at least one connected component in $H$. A connected component in
 $H$ \emph{joins} a sub-forest of $G_1$, if it joins a tree in this sub-forest. A subtree
 of $G_1$ is $\alpha$-\emph{approximate} $(0\leq \alpha\leq 1)$, if it has at least $\alpha$ times as many internal vertices as  the edges those connected components which join it have. A sub-forest of $G_1$ is
 $\alpha$-approximate $(0\leq\alpha\leq 1)$, if all  trees in it are $\alpha$-approximate.

 %This  subsection aims to  assemble the connected components in $H$ into a $\frac{3}{4}$-approximate spanning tree of $G_1$.

%To meet this aim, we show that the connected components in $H$ can be
%  assembled into such a  $\frac{3}{4}$-approximate spanning forest of $G_1$  that every connected component
%  in $H$ joins just one tree of it. %has so many internal vertices as at least two thirds of the number of edges of the components which participate the subtree.
%We hope that every tree in $F$ has so many internal vertices as at least two thirds of the number of edges in the components participating in the tree.
% Thus $F$ is initialized to contain  all the long path components in $H^{'}$. %Then our algorithm assembles short path components, singletons and cycle components in turn.

To construct a $\frac{3}{4}$-approximate spanning tree of $G_1$, we first assemble those connected components in $H$ into a $\frac{3}{4}$-approximate sub-forest of $G_1$.
A path component of length at least four in $H$ is a $\frac{3}{4}$-approximate subtree of $G_1$ naturally. Thus,

\begin{lemma}\label{lemma-long-path-added}
There is  a $\frac{3}{4}$-approximate sub-forest of $G_1$, such that every path component of length at least four in $H$  joins one tree of it.
\end{lemma}

\begin{proof}
Let $F$ be the set of path components of length at least four in $H$. Then $F$ is such a forest as  the lemma states. \qed
\end{proof}

Those singletons and path components of length 1, 2 or 3 in $H$, if present, can  be  assembled together with the path components of length at least 4 respectively, thus into a $\frac{3}{4}$-approximate sub-forest with more vertices than those of the forest  Lemma \ref{lemma-long-path-added} states.

\begin{lemma}
\label{short-path-contribution}
There is  a $\frac{3}{4}$-approximate sub-forest of $G_1$, such that every path component in $H$  joins one tree of it.
\end{lemma}

\begin{proof}
By Lemma \ref{lemma-long-path-added}, let $F$ be the $\frac{3}{4}$-approximation sub-forest formed by the path components of length at least 4 in $H$. Note that for an arbitrary subtree, say $T$ of $G_1$,  $E(H[V(T)])$ represents  the set of edges  the connected components which join $T$ have.
If $p$ is a  path component of length 1, 2 or 3 in $H$, then by Lemma \ref{property-of-path-less-than-four}, there is an edge, say $e$, in $G_1$ which is incident with  an endpoint of $p$ and a vertex of a tree, say $q$ in $F$. Adding  $e$ between  $p$ and $q$  will merge $p$ and $q$ into a new tree with the vertex set $V(p)$ $\cup$ $V(q)$. This tree  must be $\frac{3}{4}$-approximate, because it has at least $|E(H[V(p)])|$ $+$ $\frac{3}{4}|E(H[V(q)])|$ internal vertices, while the  connected components joining it by all, have  $|E(H[V(p)])|$ $+$ $|E(H[V(q)])|$ edges. By this method, every  path component of length 1, 2 or 3 can be made to join one $\frac{3}{4}$-approximate subtree. This must give rise to a $\frac{3}{4}$-approximate sub-forest of $G_1$ joined by all  path components but singletons, which will be denoted as $F$ insistently.

If $p$ is a singleton in $H$, by Lemma \ref{property-of-path-less-than-four}, there is an edge, say $e$, in $G_1$ which is incident with   $p$ and an internal vertex of a tree, say $q$ in $F$. Then $p$ and $q$ can be merged into a new tree by adding $e$ between them. With the vertex set $V(p)$ $\cup$ $V(q)$, this tree is $\frac{3}{4}$-approximate, because it has at least  $\frac{3}{4}|E(H[V(q)])|$ internal vertices, while the connected components joining it by all, have $|E(H[V(q)])|$ edges. By this way, every singleton can be made to join  one $\frac{3}{4}$-approximate subtrees. This must give rise to a $\frac{3}{4}$-approximate sub-forest of $G_1$ joined by all   path components.
\qed
\end{proof}

The remainder is  to construct a  $\frac{3}{4}$-approximate forest such that all connected components in $H$  join it.

%\textbf{Stage 3:} For the set of cycle components in $S_2$, if there is a cycle component $c$ in $S_2$, a vertex $u\in V(F)$ and a vertex $v\in V(c)$ such that $u$ is $G_{1}$-adjacent to $v$, then add $c$ to $F$, add the edge $(u,v)$, and delete one of the edges incident on $v$ in $E(c)$.

\begin{lemma}
\label{cycle-contribution}
There is a $\frac{3}{4}$-approximate spanning forest of $G_1$, such that every connected component in $H$ joins one tree of it.
\end{lemma}

\begin{proof}
By Lemma \ref{short-path-contribution}, let $F$ be a  $\frac{3}{4}$-approximate sub-forest of $G_1$ which is joined by all  path components in $H$.
If $p$ is a cycle component in $H$, then since $G_1$ is connected, there exists an edge of $G_1$ which is incident with  a vertex of $p$ and either a vertex of a tree in $F$ or a vertex of a cycle component other than $p$.

(1)If there is an edge of $G_1$ which is incident with a vertex, say $u$ of $p$ and a vertex  in a tree $T$ in $F$, then $T$ and $p$ can be assembled  into a tree, say $T'$,  by adding this edge between them, and deleting an edge of $p$ incident with $u$. Since $p$ has at least $4$ vertices,  $T'$ must be $\frac{3}{4}$-approximate, because it has at least $\frac{3}{4}|E(p)|$ + $\frac{3}{4}|E(H[V(T)])|$ internal vertices, while the connected components joining it by all, have $|E(p)|$ + $|E(H[V(T)])|$ edges. Removing $T$ from and appending $T'$ to $F$ must result in a  $\frac{3}{4}$-approximate sub-forest which is joined by more components than those joining $F$.

(2)If there is an edge in $G_1$ which is incident with a vertex, say $u$ of  $p$ and a vertex, say  $v$ of another cycle component, say $q$ in $H$, then $p$ and $q$ can be assembled into a tree, say $T'$, by adding  $(u,v)$ between them, and deleting an edge of $p$ incident with $u$ and an edge of $q$ incident with $v$. Since both $p$ and $q$ have at least $4$ vertices, $T'$ must be $\frac{3}{4}$-approximate, because it has at least $\frac{3}{4}|E(p)|$ + $\frac{3}{4}|E(q)|$ internal vertices, while $p$ and $q$ together have $|E(p)|$ + $|E(q)|$ edges. That being the case, appending  this tree to $F$  must result in a  $\frac{3}{4}$-approximate sub-forest which is joined by more components than those joining $F$.

If we insist using $F$ to represent that $\frac{3}{4}$-approximate sub-forest of $G_1$ resulted by the method of $(1)$ or $(2)$, then by the methods of (1) and (2) repeatedly, all  cycle components in $H$ can  be made to join the subtrees in $F$, which keeps to be  $\frac{3}{4}$-approximate all the time.

When all connected components in $H$ are made to join one tree in $F$, then $F$ is a spanning forest  of $G_1$, because $V(F)$ = $V(H)$ = $V(G)$ at this time.
\qed
\end{proof}

 Since $G_1$ is connected,  we can use a set of edges of $G_1$ to link all trees in the forest made by Lemma \ref{cycle-contribution} into  a spanning tree of $G_1$, which has no less  internal vertices than all those trees in the  forest have.

Finally, we integrate those computational steps for finding a spanning tree of a reduced graph into an algorithm named as SpanningTree$(G_1)$, where $G_1$ stands for an arbitrary reduced graph. In this algorithm, by reconstructing $H$, Reconstruct$(G_1,H)$ returns a maximum path-cycle cover of $G_1$ which subjects to Lemma \ref{property-of-path-less-than-four}.

\begin{algorithm}
\caption{SpanningTree$(G_1)$.}
\label{MISTG2}
    \begin{algorithmic}[1]
        \REQUIRE~~\\
            $G_1$: a reduced  graph. \\
        \ENSURE~~\\
            A spanning tree of $G_1$.\\
        \STATE Find a maximum path-cycle cover $H$ of $G_{1}$;\\
        \label{alg-maximum-path-ccyle-cover}
        \STATE $H'$ $\leftarrow$ Reconstruct$(G_1,H)$; (Lemma \ref{property-of-path-less-than-four})\\
        \label{Find_MPCC}
       \STATE $F$ $\leftarrow$ \{$p$  $\in$ $H^{'}$, $p$ is a path component, $|E(p)|>{3}$\}; (Lemma \ref{lemma-long-path-added})\\
       \label{assemble-begin}
       \STATE Assemble all  path components in $H'$  into $F$; (Lemma \ref{short-path-contribution})\\
            \label{assemble-path}
       \STATE Assemble all  cycle components in $H'$ into $F$; (Lemma \ref{cycle-contribution})\\
       \label{assemble-end}
       \STATE Link the trees in $F$ into a spanning tree $T$;
        \STATE Return $T$.
        \label{link-process}
    \end{algorithmic}

\end{algorithm}

\begin{lemma}
\label{G2-approximation}
The algorithm  SpanningTree$(G_1)$ must return a  spanning tree of $G_1$ which has $\frac{3}{4}$ times as many internal vertices as those a maximum internal spanning tree of $G_1$ has.
\end{lemma}

\begin{proof}
Let $T$ be the tree returned  by  SpanningTree$(G_1)$. By Lemma \ref{cycle-contribution},  that spanning forest of $G_1$ from which $T$ is made is $\frac{3}{4}$-approximate. Moreover, linking a spanning forest  into a tree does not add any extra vertex to that tree and loss any internal vertex of that forest. Thus $T$ is $\frac{3}{4}$-approximate. By Theorem \ref{internal-bound}, the proof is done.
  \qed
\end{proof}

Let $G_1$ = $(V_{1}$, $E_{1})$.  It takes $O(|V_1|$ $|E_1|^{1.5}$ $log|V_1|)$ time to find a maximum path-cycle cover of $G_{1}$ \cite{Edmonds1970}; it takes $O(|V_1|+|E_1|)$ time  to reconstruct a  maximum path-cycle cover each of whose cycle components has at least four edges. Thus, Step \ref{alg-maximum-path-ccyle-cover}, \ref{Find_MPCC} of SpanningTree$(G_1)$ takes $O(|V_1|$ $|E_1|^{1.5}$ $log|V_1|)$ time. Each step from \ref{assemble-begin} to \ref{assemble-end} for assembling those connected components into a spanning forest of $G_1$ takes $O(|V_1|+|E_1|)$ time. To sum up, the time complexity of SpanningTree$(G_1)$ is $O(|V_1|$ $|E_1|^{1.5}$ $log|V_1|)$.

\begin{theorem}
For any undirected simple graph, MIST can be approximated to a performance ratio $\frac{4}{3}$ in polynomial time.
\end{theorem}

\begin{proof}
Recalling to Section \ref{pre-process}, if MIST can be approximated to $\frac{4}{3}$ on reduced graphs in polynomial time, it  can also be approximated to $\frac{4}{3}$ on all graphs in polynomial time. By  Lemma \ref{G2-approximation}, the proof is done. \qed
\end{proof}

A spanning tree of an arbitrary undirected simple graph can be found by first calling Reduce$(\bullet)$ to get a reduced graph, then calling SpanningTree$(\bullet)$ to get a spanning tree of that reduced graph, and finally readding those leaves deleted by Reduce$(\bullet)$ to the tree. The time complexity for finding such a panning tree of  $G$ $=$ $(V,E)$ is $O(|V|$ $|E|^{1.5}$ $log|V|)$.

\subsection{An example}
\label{tightness-34-algorithm}
In Fig. \ref{MISTnoleaf}, we give an example to verify the performance of the algorithm. The $\frac{4}{3}$-approximation algorithm starts with a maximum path-cycle cover exactly containing $k$ cycles of length 4. The algorithm will output $T$ as its solution, while  $T^{*}$ is a spanning tree as  an optimal solution.  Since $T$ has  $3k$ internal vertices, while $T^{*}$ has  $4k-2$ internal vertices, thus increasing $k$, we come close to a $\frac{4}{3}$ ratio.

\begin{figure}[htbp]
\begin{center}
 \includegraphics[width=0.8\textwidth]{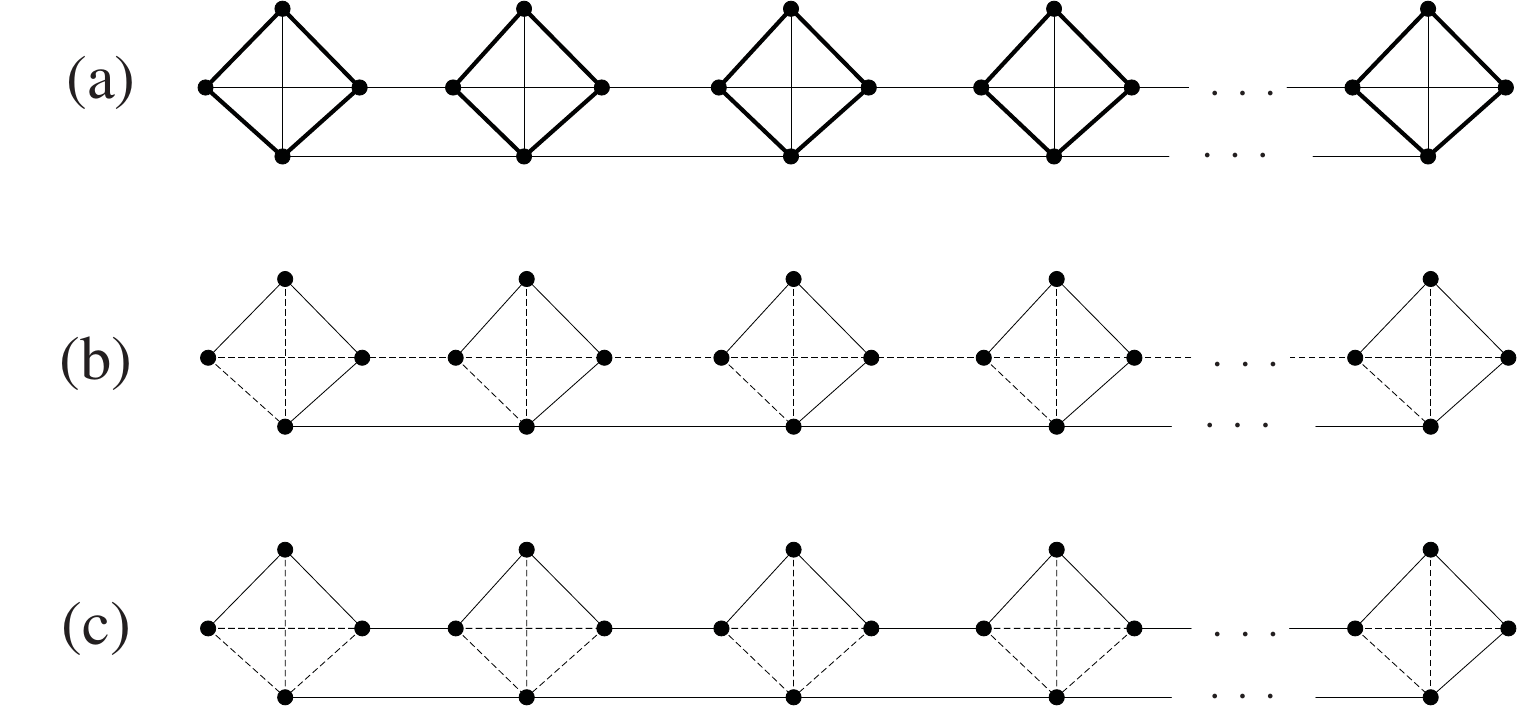}
\caption{\label{MISTnoleaf}$(a)$A graph $G$ whose maximum path-cycle cover are composed of $k$ cycles of length 4 (bold line squares). $(b)$The spanning tree of $3k$ internal vertices as a solution of SpanningTree$(G)$, if it starts with a path-cycle cover of $G'$ as in (a) of this figure. $(c)$A  maximum internal spanning tree of $G$ with $4k-2$ internal vertices. }
\end{center}
  \end{figure}

\section{Hardness to approximate MIST}
\label{hardness}
In this section, we show that if P $\neq$ NP, MIST  cannot be approximated  to within $1+\epsilon$ for some $\epsilon >0$ in polynomial time. To do this, we first present a reduction from  (1,2)-TSP  to MPC, then a   reduction from MPC to MIST. As a typical NP-hard optimization problem,
(1,2)-TSP is given by an undirected complete graph in which each edge has weight 1 or 2, and asks to find a Hamilton cycle of this graph such that the total weights of its edges  is minimized. A Hamilton cycle of an edge weighted graph is \emph{minimum weighted}, if  the total weights of its edges  is minimized over all Hamilton cycles of that graph. MPC has been proven Max-SNP-hard in \cite{Papa1993}.

\begin{theorem}
\label{hardness-of-maximum-path-cover}
If P $\neq$ NP, then for some $\epsilon >0$,  MPC cannot be approximated to within $1+\epsilon$  in polynomial time.
\end{theorem}
\begin{proof}
The proof is a reduction from (1,2)-TSP.
Let $G$ be a graph as an instance of (1,2)-TSP. We set a graph, say $G'$,  as an instance of the Maximum Path Cover problem by deleting all those edges of weight 2 from $G$.
Let $C^*$ be a minimum weighted Hamilton cycle of $G$,  $P^*$  a maximum path cover of $G'$. For a Hamilton cycle of $G$, say $C$, we denote by $w(C)$ the total weights of the edges of $C$. Let $|V(G)|$ = $|V(G')|$ = $n$.

\begin{property}
\label{L-reduction-OPT-111}
If $G'$ has a Hamilton cycle, then $w(C^*)$ + $|E(P^*)|$ = $2n-1$. Otherwise, $w(C^*)$ + $|E(P^*)|$ = $2n$.
\end{property}
\begin{proof}
If $G'$ has a Hamilton cycle, then $w(C^*)$ = $n$, $|E(P^*)|$ = $n-1$. Thus $w(C^*)$ + $|E(P^*)|$ = $2n-1$. Otherwise, $C^*$ must have $n$- $|E(P^*)|$ edges of weight 2. Thus, $w(C^*)$ + $|E(P^*)|$ = $2n$. \qed
\end{proof}

\begin{property}
 \label{L-reduction-solution}
If $G'$ has a path cover, say  $P$,  with at least 2 path components, then $G$ has a Hamilton cycle, say $C$,  with
    $w(C)-w(C^{*})$ $\leq$ $2(|E(P^{*})|-|E(P)|)$.
\end{property}

\begin{proof}
Let $p_1$, ..., $p_k$ be the path components in $P$.
Let $C$ be the Hamilton cycle by adding the edge of $G$ between an endpoint of $p_i$ and an endpoint of $p_{i+1}$ for  $1$ $\leq$ $i$ $\leq$ $k$, where $p_{k+1}$ = $p_1$.
Then $w(C)$ $\leq$ $|E(P)|$ $+$ $2(n-|E(P)|)$ $=$ $2n-|E(P)|$.  If $G'$ has no Hamilton cycle, then  by Property \ref{L-reduction-OPT-111}, $w(C)$ $\leq$ $w(C^{*})$ $+$ $|E(P^{*})|$ $-$ $|E(P)|$, and $w(C)-w(C^{*})$ $\leq$ $|E(P^{*})|$ - $|E(P)|$ consequently. If $G'$ has a Hamilton cycle, then  by Property \ref{L-reduction-OPT-111},  $w(C)$ $\leq$ $w(C^{*})$ $+$ $|E(P^{*})|$ $+$ $1-|E(P)|$. Since $|P|$ $\geq$ $2$, $|E(P^{*})|$ $-$ $|E(P)|$ $\geq$ $1$. Thus, $w(C)$ $-$ $w(C^{*})$ $\leq$ $2(|E(P^{*})|-|E(P)|)$. \qed
\end{proof}

If for some $\epsilon$ $>$ $0$, MPC can be approximated to $1+\epsilon$ in polynomial time, we argue that $(1,2)$-TSP
 can  be approximated to within $(1+2\epsilon)$  in polynomial time, which contradicts to the fact that (1,2)-TSP is Max-SNP-Hard \cite{Papa1993}. Let $P$ be a path cover as a solution of that $(1+\epsilon)$-approximation algorithm  for $G'$. Then,
 $|E(P^{*})|-|E(P)|$ $\leq$ $\epsilon |E(P)|$. Let $C$ be a Hamilton cycle constructed from $P$ by adding edges between the endpoints of those path components in $P$.

If $P$ has only one path component, then  $w(C)$ $\leq$ $w(C^*)$ + 1 $\leq$ $(1+\frac{1}{n})$ $w(C^*)$. If $n$ $\geq$ $\frac{1}{\epsilon}$, then $w(C)$ $\leq$ $(1+\epsilon)$ $w(C^*)$. If $n$ $<$ $\frac{1}{\epsilon}$, we can enumerate at most $O(n^{\frac{1}{\epsilon}})$ Hamilton cycles of $G$ to find a minimum weighted one.

If $P$ has at least 2 path components, then $C$ is a $1+2\epsilon$ solution of $G$ as a $(1,2)$-TSP instance, because  by Property  \ref{L-reduction-solution}, $w(C)$ $\leq$ $w(C^*)$ $+$ $2(E(P^{*})-|E(p)|)$ $\leq$ $w(C^*)$ + $2\epsilon |E(P)|$ $\leq$ $w(C^*)$ + $2\epsilon|E(P^{*})|$ $\leq$ $w(C^*)$ + $2\epsilon{w(C^{*})}$ $\leq$ $ (1+2\epsilon)w(C^{*})$.
\qed
\end{proof}

%%%%%%%%

\begin{theorem}
If $P\not =NP$, then for some $\epsilon>0$, MIST cannot be  approximated to within $1+\epsilon$  in polynomial time.
\end{theorem}

\begin{proof}
 The reduction is from MPC. Let $G$ be a graph as an instance of the Maximum Path Cover problem. We construct a graph $G'$ as an instance of MIST, by introducing a new vertex and connecting  it with each vertex of $G$ by an edge. Concretely, let $G$ $=$ $(V,E)$, then  $V(G')$ $=$ $V\cup\{v\}$,  $E(G')$ $=$ $E$ $\cup$ $\{(v,u):$ $u\in V\}$, where $v$ $\notin$ $V$.
 %Let $P^*$ be a maximum path cover of $G$, $T^{*}$ a maximum internal spanning tree of $G'$.
Let $P^*$ be a maximum path cover of $G$. Then, $|P^*|$ + $|E(P^*)|$ = $|V(G)|$.
Let $T^{*}$ be a maximum internal spanning tree of $G'$, $I(T^*)$ the set of internal vertices of $T^*$, $L(T^{*})$ the set  of leaves of $T^{*}$. Then $|L(T^*)|$ + $|I(T^*)|$ $=$ $|V(G)|$ +1.

\begin{property}
\label{path-num-equal-leaf-num}
If $G$ has no Hamilton path, then a maximum path cover of $G$ has as many path components as the leaves a maximum internal spanning tree of $G'$  has.
\end{property}

\begin{proof}

(1) From $P^*$, we can construct a spanning tree $T$ of $G'$ by connecting $v$ with exactly one endpoint of each path component in $P^*$. Then $T$ has just $|P^{*}|$ leaves. So $|L(T^{*})|$ $\leq$ $|P^{*}|$.

(2) By Lemma \ref{tree-path-number}, there is a path cover, say $P$ of $T^*$, with less path components than the leaves of $T^*$.  That is, $|P|$ $\leq$ $|L(T^{*})|$ $-$ $1$. We can get a path cover $P'$ of $G$ by deleting  $v$ from $P$, where $v$ $\in$ $V(G')$ $\setminus$ $V(G)$. Then $P'$ has at most  $|P|$ $+$ $1$ $\leq$ $|L(T^{*})|$ path components. Namely, $|P^{*}|$ $\leq$ $|L(T^{*})|$.

Finally, $|P^{*}|$ $=$ $|L(T^{*})|$ follows from  $(1)$ and $(2)$.
\qed
\end{proof}

\begin{property}
\label{intrenal-nodes-edges-relation}
A maximum internal spanning tree of $G'$  has  no less internal vertices than the edges those path components in $P^{*}$ has.
\end{property}

\begin{proof}
If $G$ has a Hamilton path,  then $|E(P^*)|$ $=$ $|V(G)|$ $-$ $1$, and $|I(T^{*})|$ $=$ $|V(G)|$ $-$ $1$. So $|E(P^*)|$ $=$ $|I(T^{*})|$.
If  $G$ has no Hamilton path, then by Property \ref{path-num-equal-leaf-num}, $|P^*|$ $=$ $|V(G)|$ $+$ $1$ $-$ $|I(T^{*})|$. Thus,
$|I(T^{*})|$ $=$ $|V(G)|$ $+$ $1$ $-$ $|P^*|$ = $|E(P^*)|$ $+$ $1$. That is $|E(P^*)|$ $\leq$ $|I(T^*)|$.
 \qed
\end{proof}

Suppose for some $\epsilon$ $>$ $0$, an algorithm can approximate MIST to $1$ $+$ $\epsilon$ on undirected simple graphs. Let $T$ be a spanning tree of $G'$ as a solution of this algorithm. Then  $|I(T^{*})|$ $\leq$ $(1$ $+$ $\epsilon)$ $|I(T)|$. We can  construct a path cover of $G$, say $P$, first by  the  method in the proof of Lemma \ref{tree-path-number} to get a path cover of $G'$, then deleting  $v$ $\in$ $V(G')$ $\setminus$ $V(G)$ from it. By Lemma \ref{tree-path-number}, this path cover of $G$ must have no more than $|L(T)|$ path components.  That is, $P$ has at least $|V(G)|$ $-$ $|L(T)|$ $=$ $|I(T)|-1$ edges, which means  $|I(T)|$ $\leq$ $|E(P)|$ $+$ $1$. By Property \ref{intrenal-nodes-edges-relation}, $|E(P^{*})|$ $\leq$ $|I(T^{*})|$ $\leq$ $(1+\epsilon)$ $|I(T)|$ $\leq$ $(1+\epsilon)$ $(|E(P)|+1)$.  If $|E(P)|$ $\geq$ $\frac{1+\epsilon}{\epsilon}$, then $|E(P^{*})|$ $\leq$ $(1+2\epsilon)$ $|E(P)|$,  otherwise, one can use $O(|E(G)|^{\frac{1+\epsilon}{\epsilon}})$ time to find a maximum path cover of $G$. This comes to a contradiction  to Theorem \ref{hardness-of-maximum-path-cover}. \qed
\end{proof}

\section{Conclusion}
\label{conclusion}
We have presented an algorithm for  MIST which can achieve a performance ratio $\frac{4}{3}$ on undirected simple graphs. We believe that the bound of the number of internal vertices for a spanning tree can be applied to designing useful efficient approximation algorithms for other problems such as the Minimum Leaves Spanning Tree problem.  It is interesting  whether MIST can be approximated to a better  performance ratio than $\frac{4}{3}$ on undirected simple graphs. If one want to  follow the method of this paper to arrive at a better performance ratio than $\frac{4}{3}$, it seems necessary to deal with those cycle components in a maximum path-cycle cover.  It is also interesting whether a constant can be decided to which MIST  rejects to be approximated by a polynomial time algorithm, if P $\neq$ NP.


\begin{thebibliography}{16}

    \bibitem{Edmonds1970-1} J. Edmonds, E. L. Johnson,
    Matching: a well solved class of integer linear programs.
    Combinatorial Structures and their Applications, Gordon and Breach, New York, (1970) 89-92.

    \bibitem{Hartvigsen1984} D. Hartvigsen:
     Extensions of matching theory, Ph.D. Thesis.
    Carnegie-Mellon University, (1984)

    \bibitem{Edmonds1970} Y. Shiloach:
     Another look at the degree constrained subgraph problem.
    Inf. Process. Lett. {\bfseries 12(2)} (1981) 89--92


    \bibitem{Prieto2003} E. Prieto, C. Sliper.
    Either/or: Using vertex cover structure in designing fpt-algorithms--the case of k-internal spanning tree.
    WADS 2003, LNCS {\bfseries 2748} (2003) 474--483

    \bibitem{Salamon2008} G. Salamon, G. Wiener.
    On finding spanning trees with few leaves.
    Information Processing Letters, {\bfseries 105(5)} (2008) 164--169

    \bibitem{Salamon2009} G. Salamon
    Approximating the Maximum Internal Spanning Tree problem.
    Theoretical Computer Scientc {\bfseries 410(50)} (2009) 5273--5284

    \bibitem{Aalamon2009PhD} G. Salamon
    Degree-Based Spanning Tree Optimization.
    Ph.D. thesis, Budapest University of Technology and Ecnomics, Hungary. (2009)

    \bibitem{Knauer2009} M. Knauer, J. Spoerhase.
    Better Approximation Algorithms for the Maximum Internal Spanning Tree Problem.
    WADS 2009, LNCS {\bfseries 5664} (2009) 489--470


    \bibitem{Coben2010} N. Coben, F. V. Fomin, G. Gutin, E. J. Kim, S. Saurabh, A. Yeo.
    Algorithm for finding k-vertex out-trees and its application to k-internal out-branching problem.
    JCSS {\bfseries 76(7)} (2010) 650--662

    \bibitem{Fomin2010} F. V. Fomin, D. Lokshtanov, F. Grandoni, S. Saurabh.
    Sharp seperation and applications to exact and parameterized algorithms.
    Algorithmica, {\bfseries 63(3)},  692--706 (2012)

    \bibitem{Fomin2013} F. V. Fomin, S. Gaspers, S. Saurabh, S. Thomass$\acute{e}$.
    A linear vertex kernel for maximum internal spanning tree.
    JCSS {\bfseries 79} (2013) 1--6


    \bibitem{Fernau2013} D. Binkele-Raible, H. Fernau, S. Gaspers, M. Liedloff.
    Exact and parameterized algorithms for MAX INTERNAL SPANNING TREE.
    Algorithmica {\bfseries 65} (2013) 95--128

    \bibitem{Prieto2005} E. Prieto.
    Systematic kernelization in FPT algorithm design.
    Ph.D. Thesis, The University of Newcastle, Australia (2005)

    \bibitem{Prieto2005-2} E. Prieto, C. Sloper.
    Reducing to independent set structure¡ªthe case of k-internal spanning tree.
    Nord. J. Comput. {\bfseries 12(3)} (2005) 308--318

  %  \bibitem{Akiyama1980}T. Akiyama, T. Nishizeki, N Saito.
  %  NP-completeness of the Hamiltonian cycle problem for bipartite graphs.
  %  Journal of Information Processing {\bfseries 3 (2)} (1980) 73--76

     \bibitem{Garey1979}M. R. Garey and D. S. Johnson:
    Computers and Intractability: A Guide to the Theory of NP-Completeness.
    W. H. Freeman. (1979)

    \bibitem{Zehavi2013} M. Zehavi.
    Algorithms for k-Internal Out-Branching.
    Parameterized and Exact Computation {\bfseries 8246} (2013) 361--373


    \bibitem{Papa1993} C. H. Papadimitriou, M. Yannakakis.
    The travaling salesman problem with distance one and two.
    Math. Oper. Res.{\bfseries 18(1)} 1-11 1993

\end{thebibliography}
\end{document}